\documentclass[11pt]{article}

\usepackage{ifpdf}

\usepackage{amssymb}

\usepackage{pstricks, pst-coil, pst-node, pst-tree, multido}

\newtheorem{DE}{Definition}[section]

\newcommand {\sm} {\setminus}

\usepackage{graphicx}
\usepackage{latexsym} 
\usepackage{theorem} 
\newcommand{\qed}{\relax\ifmmode\hskip2em\Box\else\unskip\nobreak\hfill$\Box$\fi}

\newtheorem{theorem}[DE]{Theorem}
\newtheorem{lemma}[DE]{Lemma}

\newtheorem{CL}[DE]{Claim}
{\theoremstyle{break}\theorembodyfont{\rmfamily}}
{\theoremstyle{break}\theorembodyfont{\rmfamily}}

\newcommand{\tp}{\!-\!}

\newcounter{claim}

\newenvironment{claim}[1][]%
{\refstepcounter{claim}\vspace{1ex}\noindent{(\it\arabic{claim}){#1}{}}\it}{\vspace{1ex}}

\newenvironment{proofclaim}[1][]%
    {\noindent {}{#1}{}}{ This proves~(\arabic{claim}).\vspace{1ex}}

 \newenvironment{proof}[1][]%
 {\noindent {\setcounter{claim}{0}\sc proof ---
    }{#1}{}}{\hfill$\Box$\vspace{2ex}}

\ifpdf
\fi

\begin{document}

\title{The four-in-a-tree problem in triangle-free graphs}
\author{Nicolas Derhy\thanks{CEDRIC, CNAM, Paris (France),\ email:
    nicolas.derhy@cnam.fr}, Christophe Picouleau\thanks{CEDRIC, CNAM,
    Paris (France),\ email: chp@cnam.fr}, Nicolas
  Trotignon\thanks{CNRS, LIAFA, Universit\'e Paris 7 -- Paris Diderot (France),\ email: nicolas.trotignon@liafa.jussieu.fr}}

\date{March 10, 2009}

\maketitle

\begin{abstract}
  The three-in-a-tree algorithm of Chudnovsky and Seymour decides in
  time $O(n^4)$ whether three given vertices of a graph belong to an
  induced tree.  Here, we study four-in-a-tree for triangle-free graphs.
  We give a structural answer to the following question: what does a triangle-free graph look
  like  if no induced tree covers four
  given vertices~?  Our main result says that any such graph must have
  the ``same structure'', in a sense to be defined precisely, as a
  square or a cube.

  We provide an $O(nm)$-time algorithm that given a triangle-free
  graph~$G$ together with four vertices outputs either an induced tree
  that contains them or a partition of $V(G)$ certifying that no such
  tree exists.  We prove that the problem of deciding whether there exists
  a tree $T$ covering the four vertices such that at most one vertex of
  $T$ has degree at least~3 is NP-complete.
\end{abstract}



\section{Introduction}

Many interesting classes of graphs are defined by forbidding induced
subgraphs, see~\cite{chudnovsky.seymour:excluding} for a survey.  This
is why the detection of several kinds of induced subgraphs is
interesting, see~\cite{leveque.lmt:detect} for a survey.  In
particular, the problem of deciding whether a graph $G$ contains as an
induced subgraph some graph obtained after possibly subdividing
prescribed edges of a prescribed graph $H$ has been studied.  It
turned out that this problem can be polynomial or NP-complete
according to $H$ and to the set of edges that can be subdivided.
Details, examples and open problems are given
in~\cite{leveque.lmt:detect}.  The most general tool for solving this
kind of problems (when they are polynomial) seems to be the
\emph{three-in-a-tree} algorithm of Chudnovsky and Seymour:

\begin{theorem}[see \cite{chudnovsky.seymour:theta}]
  Let $G$ be a connected graph and $x_1, x_2, x_3$ be three distinct
  vertices of $G$. Then deciding if there exists an induced tree of
  $G$ that contains $x_1, x_2, x_3$ can be done in time $O(n^4)$.
\end{theorem}

\begin{figure}[p]
  \center
  \includegraphics{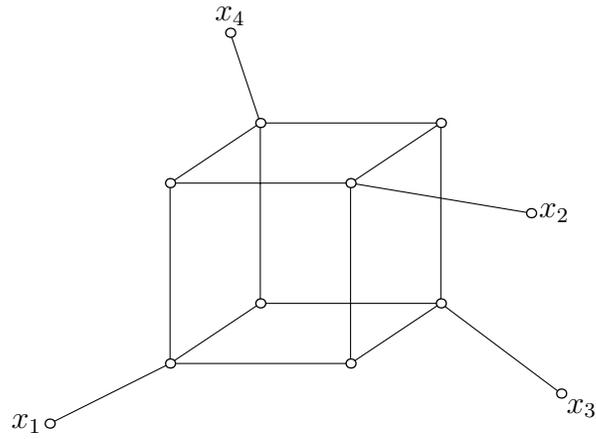}
  \caption{no tree covers $x_1, x_2, x_3, x_4$, first
    example\label{fig:cubeEx}}
\end{figure}

\begin{figure}[p]
  \center
  \includegraphics{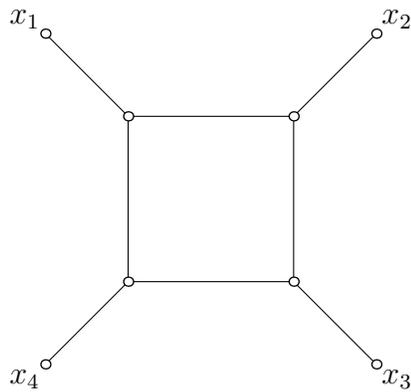}
  \caption{no tree covers $x_1, x_2, x_3, x_4$, second
    example\label{fig:squareEx}}
\end{figure}

How to use three-in-a-tree is discussed
in~\cite{chudnovsky.seymour:theta} and further evidences of its
generality are given in~\cite{leveque.lmt:detect}.  Because of the
power and deepness of three-in-a-tree, it would be interesting to
generalise it.  Here we study \emph{four-in-a-tree}: the problem whose
instance is a graph $G$ together with four of its vertices, and whose
question is ``Does $G$ contain an induced tree covering the four
vertices ?''.  Since this problem seems complicated to us, we restrict
ourselves to triangle-free graphs.  Our approach is similar to that of
Chudnovsky and Seymour for three-in-a-tree. We give a structural
answer to the following question: what does a triangle-free
graph look like if no induced tree covers four given vertices $x_1, x_2,
x_3, x_4$?  On Fig.~\ref{fig:cubeEx} and~\ref{fig:squareEx}, two
examples of such graphs are represented.  Our main result,
Theorem~\ref{th:main}, says that any triangle-free graph that does not
contain a tree covering four vertices $x_1, x_2, x_3, x_4$ must have
the ``same structure'', in a sense to be defined later, as one of the
two examples.  The details of the statement are given in
Section~\ref{s:statement}.

Our result is algorithmic: we provide an $O(nm)$-time algorithm that
given a graph~$G$ together with four vertices $x_1, x_2, x_3, x_4$
outputs either an induced tree that contains $x_1, x_2, x_3, x_4$ or a
partition of $V(G)$ certifying that no such tree exists.  Note that
apart from very basic subroutines such as Breadth First Search, our
algorithm is self-contained.  In particular it does no rely on
three-in-a-tree.  Our proofs will use the following result of
Derhy and Picouleau:

\begin{theorem}[see \cite{derhy.picouleau:threetfree}]
  Let $G$ be a triangle-free connected graph and $x_1, x_2, x_3$ be
  three distinct vertices of $G$. Then there is an induced tree of $G$
  that contains $x_1, x_2, x_3$. Moreover such a tree of minimum size
  can be done in time $O(m)$.
  \label{theo_3term}
\end{theorem}

Another generalisation of three-in-a-tree would be interesting.  Let
us call \emph{centered tree} any tree that contains at most one vertex
of degree greater than two.  Note that any minimal tree covering three
vertices of a graph is centered.  Hence, three-in-a-tree and
three-in-a-centered-tree are in fact the same problem.  So
four-in-a-centered-tree is also an interesting generalisation of
three-in-a-tree.  But we will prove in Section~\ref{s:npc} that it is
NP-complete, even when restricted to several classes of graphs,
including triangle-free graphs.

We leave open the following problems: four-in-a-tree for general
graphs, $k$-in-a-tree for triangle-free graphs.

\subsection*{Notation}

All our graphs are simple and finite.  We say that a graph $G$
\emph{contains} a graph $H$ if $G$ contains an induced subgraph
isomorphic to $H$.  We say that $G$ is $H$-free if it does not contain
$H$.  If $Z\subseteq V(G)$ then $G[Z]$ denotes the subgraph of $G$
induced by $Z$.  When we describe the complexity of an algorithm whose
input is a graph, $n$ stands for the number of its vertices and $m$
stands for the number of its edges.

We call \emph{path} any connected graph with at least one vertex of
degree~1 and no vertex of degree greater than~2. A path has at most
two vertices of degree~1, which are the \emph{ends} of the path. If
$a, b$ are the ends of a path $P$ we say that $P$ is \emph{from $a$
  to~$b$}. The other vertices are the \emph{interior} vertices of the
path. We denote by $v_1 \tp \cdots \tp v_n$ the path whose edge set is
$\{v_1v_2, \dots, v_{n-1}v_n\}$.  When $P$ is a path, we say that $P$
is \emph{a path of $G$} if $P$ is an induced subgraph of $G$. If $P$
is a path and if $a, b$ are two vertices of $P$ then we denote by $a
\tp P \tp b$ the only induced subgraph of $P$ that is path from $a$ to
$b$.

Note that by \emph{path of a graph}, we mean induced path.  Also, by
\emph{tree of a graph}, we mean an induced subgraph that is a tree.

The \emph{union} of two graphs $G= (V, E)$ and $G' = (V', E')$ is the
graph $G\cup G' = (V\cup V', E \cup E')$. A set $X \subseteq V(G)$ is
\emph{complete} to a set $Y \subseteq V(G)$ if there are all possible
edges between $X$ and $Y$. A set $X \subseteq V(G)$ is
\emph{anticomplete} to a set $Y \subseteq V(G)$ if there are no edges
between $X$ and~$Y$.

When $G$ is a graph and $v$ a vertex, $N(v)$ denotes the set of all
the neighbors of $v$. If $A \subseteq V(G)$ then $N(A)$ denotes the
set of these vertices of $G$ that are not in $A$ but that have
neighbors in $A$.  If $Z \subseteq V(G)$, then $N_Z(A)$ denotes
$N(A)\cap Z$.  If $H$ is an induced subgraph of $G$, then we write
$N_H(A)$ instead of $N_{V(H)}(A)$.

When we define $k$ sets $A_1, \dots, A_k$, we usually denote their
union by $A$. We use this with no explicit mention~: if we define sets
$S_1, \dots, S_8$ then $S$ will denote their union, and so on.

\section{Main results}
\label{s:statement}

A \emph{terminal} of a graph is a vertex of degree one. Given a graph
$G$ and vertices $y_1, \dots, y_k$, let us consider the graph $G'$
obtained from $G$ by adding for each $y_i$ a new terminal $x_i$
adjacent to $y_i$. It is easily seen that there exists an induced tree
of $G$ covering $y_1, \dots, y_k$ if and only if there exists an
induced tree of $G'$ covering $x_1, \dots, x_k$.  So,
four-terminals-in-a-tree and four-in-a-tree are essentially the same
problems, from an algorithmic point of view and from a structural
point of view.  Hence, for convenience, we may restrict ourselves to
the problem four-in-a-tree where the four vertices to be covered are
terminals.

As mentioned in the introduction, our main result states that a graph
that does not contain a tree covering four given terminals $x_1, x_2,
x_3, x_4$ must have the ``same structure'' as one of the graphs
represented on Fig.~\ref{fig:cubeEx} or \ref{fig:squareEx}.  Let us
now define this precisely.

A graph that has the same structure as the graph represented on
Fig~\ref{fig:cubeEx} is what we call a cubic structure: a graph $G$ is
said to be a \emph{cubic structure} with respect to a 4-tuple of
distinct terminals $(x_1, x_2, x_3, x_4)$ if there exist sets $A_1,
\dots A_4, $ $B_1, \dots B_4$, $S_1, \dots, S_8$ and $R$ such that:

\begin{enumerate}
  \item
    $A \cup B \cup S \cup R = V(G)$;
  \item
    $A_1, \dots, A_4, B_1, \dots, B_4, S_1, \dots, S_8, R$ are
    pairwise disjoint;
  \item
    $x_i \in A_i$, $i= 1, \dots, 4$;
  \item
    $S_i$ is a stable set, $i= 1, \dots, 8$;
  \item
    $S_i$ is non-empty, $i= 1, \dots, 4$;
  \item
    \label{i:cubeNe}
    at most one of $S_5, S_6, S_7, S_8$ is empty;
  \item
    $S_i$ is complete to $(S_5 \cup S_6 \cup S_7 \cup S_8) \sm
    S_{i+4}$, $i= 1, 2, 3, 4$;
  \item
    $S_i$ is anticomplete to $S_{i+4}$, $i= 1, 2, 3, 4$;
  \item
    $S_i$ is anticomplete to $S_j$, $1 \leq i < j \leq 4$;
  \item
    $S_i$ is anticomplete to $S_j$, $5 \leq i < j \leq 8$;
  \item
    \label{i:cubeNAi}
    $N(A_i) = S_i$, $i= 1, 2, 3, 4$;
  \item
    \label{i:cubeNBi}
    $N(B_i) \subseteq S_i \cup N_S(S_i)$,  $i= 1, 2, 3, 4$;
  \item
    \label{i:cubeNR}
    $N(R) \subseteq S_5 \cup S_6 \cup S_7 \cup S_8$;
  \item
    \label{i:cubeACon}
     $G[A_i]$ is connected, $i= 1, 2, 3, 4$.
\end{enumerate}

\begin{figure}
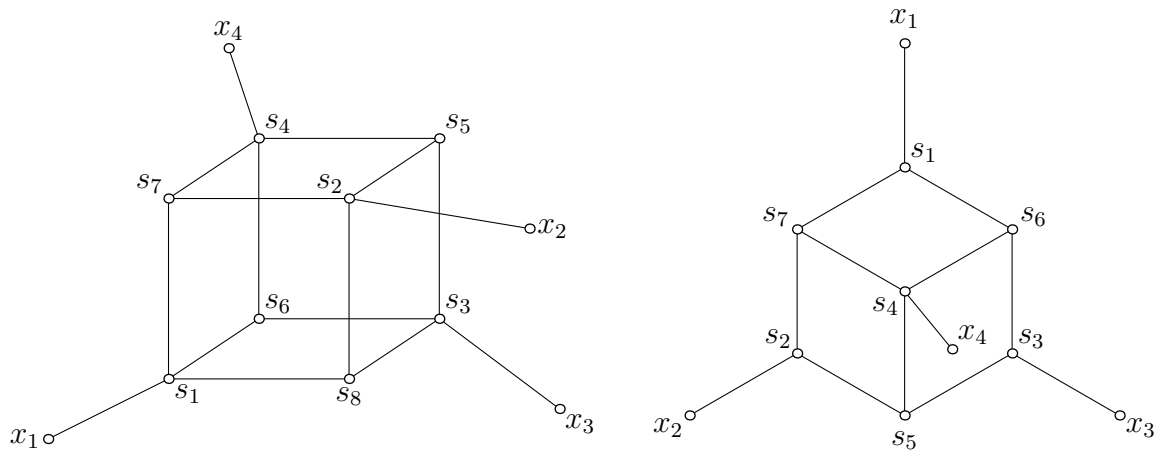

  \center
  \includegraphics{figFtree.1}
  \rule{0.5cm}{0cm}
 \includegraphics{figFtree.3}
 \caption{Two examples of cubic structure\label{fig:cubic}}
\end{figure}

A 17-tuple $(A_1, \dots A_4, $ $B_1, \dots B_4$, $S_1, \dots, S_8, R)$
of sets like in the definition above is a \emph{split} of the cubic
structure.  On Fig.~\ref{fig:cubic}, two cubic structures are
represented.  A \emph{cubic structure of a graph $G$} is a subset $Z$
of $V(G)$ such that $G[Z]$ is a cubic structure.  The following lemma,
to be proved in Section~\ref{s:cube}, shows that if a cubic structure
is discovered in a triangle-free graph, then one can repeatedly add
vertices to it, unless at some step a tree covering $x_1, x_2, x_3,
x_4$ is found:

\begin{lemma}
  \label{l:cube}
  There is an algorithm with the following specification:

  \begin{description}
  \item
    {\sc Input:} a triangle-free graph $G$, four terminals $x_1, x_2,
    x_3, x_4$, a split of a cubic structure $Z$ of $G$, and a vertex
    $v\notin Z$.
  \item
    {\sc Output:} a tree of $G[Z \cup \{v\}]$ that covers $x_1, x_2,
    x_3, x_4$ or a split of the cubic structure $G[Z \cup \{v\}]$.
  \item
    {\sc Complexity: } $O(m)$.
  \end{description}
\end{lemma}

Let us now turn our attention to our second kind of structure.  A
graph that has the same structure as the graph represented on
Fig~\ref{fig:squareEx} is what we call a square structure: a graph $G$
is said to be a \emph{square structure} with respect to a 4-tuple
$(x_1, x_2, x_3, x_4)$ of distinct terminals if there are sets $A_1,
A_2, A_3, A_4, S_1, S_2, S_3, S_4, R$ such that:

\begin{enumerate}
  \item
    $A \cup S \cup R = V(G)$;
  \item
    $A_1, A_2, A_3, A_4, S_1, S_2, S_3, S_4, R$ are pairwise disjoint;
  \item
    $x_i \in A_i$, $i= 1, \dots, 4$;
  \item
    $S_i$ is a stable set, $i= 1, \dots, 4$;
  \item
    $S_1, S_2, S_3, S_4 \neq \emptyset$;
  \item
    $S_i$ is complete to $S_{i+1}$, where the addition of subscripts is
    taken modulo 4, $i= 1, 2, 3, 4$;
  \item
    $S_i$ is anticomplete to $S_{i+2}$, $i= 1, 2$;
  \item
    \label{i:squareNAi}
    $N(A_i) = S_i$ , $i= 1, 2, 3, 4$;
  \item
    $N(R) \subseteq S_1 \cup S_2 \cup S_3 \cup S_4$;
  \item
    \label{i:squareACon}
    $G[A_i]$ is connected, $i= 1, \dots, 4$.
 \end{enumerate}

\begin{figure}
  \center
  \includegraphics{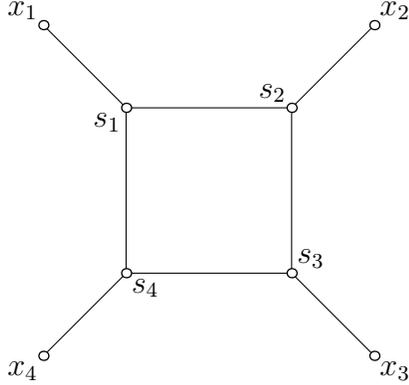}
  \caption{The smallest square structure\label{fig:square}}
\end{figure}

A 9-tuple $(A_1, \dots A_4, S_1, \dots, S_4, R)$ of sets like in the
definition above is a \emph{split} of the square structure.  On
Fig.~\ref{fig:square}, the smallest square structure is represented.
A \emph{square structure of a graph $G$} is a subset $Z$ of $V(G)$
such that $G[Z]$ is a square structure. The following lemma, to be
proved in Section~\ref{s:square}, shows that if a square structure is
discovered in a triangle-free graph, then one can repeatedly add
vertices to it, unless at some step a cubic structure or a tree
covering $x_1, x_2, x_3, x_4$ is found:

\begin{lemma}
  \label{l:square}
  There is an algorithm with the following specification:

  \begin{description}
  \item
    {\sc Input:} a triangle-free graph $G$, four terminals $x_1, x_2,
    x_3, x_4$, a split of a square structure $Z$ of $G$, and a vertex
    $v \notin Z$.
  \item
    {\sc Output:} a tree of $G[Z \cup \{v\}]$ that covers $x_1, x_2,
    x_3, x_4$ or a split of some cubic structure of $G$ or a split of
    the square structure $G[Z \cup \{v\}]$.
  \item
    {\sc Complexity: } $O(m)$.
  \end{description}
\end{lemma}

From the two lemmas above the main theorem follows:

\begin{theorem}
  \label{th:main}
  Let $G$ be a connected graph and $x_1, x_2, x_3, x_4$ be four
  distincts terminals of $G$. Then either:
  \begin{itemize}
  \item[(1)]
    $G$ is a cubic or a square structure with respect to $(x_1, x_2,
    x_3, x_4)$;
  \item[(2)]
    $G$ contains a tree that covers $x_1, x_2, x_3, x_4$.
  \end{itemize}

  Moreover, exactly one of these two statements (1) and (2) holds.  This
  result is algorithmic in the sense that there exists an $O(nm)$-time
  algorithm whose input is a graph and four terminals $x_1, x_2, x_3,
  x_4$ and whose output is either a partition of $V(G)$ showing that
  $G$ is a cubic or a square structure with respect to $(x_1, x_2,
  x_3, x_4)$, or a tree that covers $x_1, x_2, x_3, x_4$.
\end{theorem}

\begin{proof}
  Let us first check that at most one of (1), (2) holds.  This means
  that a square or cubic structure with respect to a 4-tuple $(x_1,
  x_2, x_3, x_4)$ of terminals cannot contain a tree covering $x_1,
  x_2, x_3, x_4$.  For suppose that such a tree $T$ exists in a square
  structure $G$ with a split $(A_1, \dots, A_4, S_1, \dots, S_4, R)$.
  By the definition of square structures $T$ must contain a vertex in
  every $S_i$, $i= 1, 2, 3, 4$.  So, $T$ contains a square, a
  contradiction.

  Suppose now that such a tree $T$ exists in a cubic structure $G$
  with a split $(A_1, \dots, A_4, B_1, \dots, B_4, S_1, \dots, S_8,
  R)$.  By the definition of cubic structures $T$ must contain a
  vertex in every $S_i$, $i= 1, 2, 3, 4$.  Since $T$ contains no
  cycle, $T$ has vertices in at most one of $S_5, S_6, S_7, S_8$, say
  in $S_5$ up to symmetry.  So, $T$ contains no vertex of $S_6 \cup
  S_7 \cup S_8$.  So $x_1, x_2$ lie in two different components of
  $T$, a contradiction.

  The fact that at least one of (1), (2) holds follows directly from
  the algorithm announced in the theorem.  A description of this
  algorithm will complete the proof. So let us suppose that $G$ and
  four terminals $x_1, x_2, x_3, x_4$ are given. The algorithm goes
  through three steps:

  \vspace{1ex}

  \noindent {\bf First step:} by Theorem~\ref{theo_3term} we find in
  time $O(m)$ a minimal tree $T$ that covers $x_1, x_2, x_3$.  Note that since $x_1, x_2, x_3$ are of
  degree one, $T$ contains a vertex $c$ of degree~$3$ and is the union
  of three paths $P_1 = c \tp \cdots \tp x_1$, $P_2 = c \tp \cdots \tp
  x_2$, $P_3 = c \tp \cdots \tp x_3$.

  Then we use BFS (short name for Breadth First Search,
  see~\cite{gibbons:agt}) to find a path $Q = x_4 \tp \cdots \tp w$
  such that $w$ has neighbors in $T$, and minimal with respect to this
  property.  If $w$ has a neighbor in $P_i$ then we let $u_i$ be the
  neighbor of $w$ closest to $x_i$ along $P_i$.

  If $w$ has neighbors in $P_1, P_2, P_3$ then note that when $1\leq i
  < j \leq 3$, $u_iu_j \notin E(G)$ because else, $G$ contains a
  triangle.  So $V(Q \cup (u_1 \tp P_1 \tp x_1) \cup (u_2 \tp P_2 \tp
  x_2) \cup (u_3 \tp P_3 \tp x_3))$ induces a tree that covers $x_1, x_2,
  x_3, x_4$, so we stop the algorithm and output this tree.  Note that
  from here on, $wc\notin E(G)$.

  If $w$ has neighbors in exactly one of $P_1, P_2, P_3$, say in $P_1$
  up to symmetry, then we compute by Theorem~\ref{theo_3term} a tree $T'$
  of $G[P_1 \cup \{w\}]$ that minimally covers $w, c, x_1$. We see that
  $V(Q \cup T' \cup P_2 \cup P_3)$ induces a tree that covers $x_1,
  x_2, x_3, x_4$, so we stop the algorithm and output this tree.

  So, we are left with the case where $w$ has neighbors in two paths
  among $P_1, P_2, P_3$, say in $P_1, P_3$ up to symmetry.  Then there
  are two cases.  First case: one of $u_1c, u_3c$ is not in $E(G)$.
  Up to symmetry we suppose $u_1c \notin E(G)$.  We compute by
  Theorem~\ref{theo_3term} a tree $T''$ of $G[P_3 \cup \{w\}]$ that
  minimally covers $w, c, x_3$. We see that $V(Q \cup (x_1 \tp P_1 \tp
  u_1) \cup T'' \cup P_2)$ induces a tree that covers $x_1, x_2, x_3,
  x_4$, so we stop the algorithm and output this tree.  Second case:
  $u_1c, u_3c$ are both in $E(G)$.  Then we observe that $V(P_1 \cup
  P_2 \cup P_3 \cup Q)$ is a square structure of~$G$. A split can be
  done by putting $A_1 = V(x_1 \tp P_1 \tp u_1) \sm \{u_1\}$, $A_2 = V(P_2)\sm \{c\}$, $A_3 =
  V(x_3 \tp P_3 \tp u_3) \sm \{u_3\}$; $A_4 =
  V(Q) \sm \{w\}$, $S_1 = \{u_1\}$, $S_2 = \{c\}$, $S_3 = \{u_3\}$,
  $S_4 = \{w\}$ and $R = \emptyset$.  We keep this square structure $Z$
  and go the next step.

 \vspace{1ex}

  \noindent {\bf Second step:} while there exists a vertex $v$ not in
  $Z$, we use the algorithm of Lemma~\ref{l:square} to add $v$ to $Z$,
  keeping a square structure.  If we manage to put every vertex of $G$
  in $Z$ then we have found that $G$ is a square structure that we
  output.  Else, Lemma~\ref{l:square} says that at some step we have
  found either a tree covering $x_1, x_2, x_3, x_4$ that we output, or
  a cubic structure $Z'$, together with a split for it. In this last
  case, we go to the next step.

 \vspace{1ex}

  \noindent {\bf Third step:} while there exists a vertex $v$ not in
  $Z'$, we use the algorithm of Lemma~\ref{l:cube} to add $v$ to $Z'$,
  keeping a cubic structure.  If we manage to put every vertex of $G$
  in $Z'$ then we have found that $G$ is a cubic structure that we
  output.  Else, Lemma~\ref{l:cube} says that at some step we have
  found a tree covering $x_1, x_2, x_3, x_4$ that we output.

 \vspace{1ex}

  \noindent {\bf Complexity analysis:} we run at most $O(n)$ times $O(m)$
  algorithms. So the overall complexity is $O(nm)$.
\end{proof}

\section{Proof of Lemma~\ref{l:square}}
\label{s:square}

  Let $Z \subseteq V(G)$ be a square structure of $G$ with respect to
  $x_1, x_2, x_3, x_4$ together with a split like in the definition
  and let $v$ be in $V(G) \sm Z$.  Note that $Z$, the split of $Z$ and
  $v$ are given by assumption.

  Here below, we give a proof of the existence of the objects that the
  algorithm of our Lemma must output, namely a tree, a cubic structure
  or an augmented square structure.  But this proof is in fact the
  description of an $O(m)$-time algorithm.  To see this, it suffices
  to notice that all the proof relies on a several run of BFS or of
  the algorithm of Theorem~\ref{theo_3term}, and on checks of
  neighborhoods of several vertices.  At the end, we give more
  information on how to transform our proof into an algorithm.

  When $s_i \in S_i\cup A_i$, we define the path $P_{s_i}$ to be a
  path from $s_i$ to $x_i$, whose interior is in $A_i$, $i= 1, 2, 3,
  4$.  Note that $P_{s_i}$ exists since by Item~\ref{i:squareACon} of
  the definition of square structures, $G[A_i]$ is connected and by
  Item~\ref{i:squareNAi} every vertex of $S_i$ has a neighbor in
  $A_i$, $i=1, 2, 3, 4$.

  If $v$ has no neighbor in $A$ then $v$ can be put in $R$ and we
  obtain a split of the square structure $Z\cup \{v\}$.  So we may
  assume that $v$ has a neighbor in $A$, say $a_1 \in A_1$.  We choose
  $a_1$ subject to the minimality of $P_{a_1}$.

  \begin{CL}
    \label{c:squareA4}
    Suppose that there exists a path $Q= v \tp \cdots \tp w$ where
    $Q\sm v \subseteq R$ and such that $w$ has neighbors in $(A\sm
    A_1) \cup (S\sm S_1)$. Suppose $Q$ minimal with respect to these
    properties.  Then either:
    \begin{enumerate}
      \item
        there exists a tree of $G[Z \cup \{v\}]$ that covers
        $x_1,$ $x_2,$ $x_3,$ $x_4$;
      \item
        $N_{S} (w) = S_2 \cup S_4$ and $N_A(w) \subseteq A_1$;
      \item
        $G[Z \cup \{v\}]$ contains a cubic structure.
    \end{enumerate}
  \end{CL}

  \begin{proof}
    Note that possibly $Q=v=w$.  Note also that by the definition of
    $R$, $w$ have neighbors in $A$ only when $w=v$.

    \begin{claim}
     If $w$ has a neighbor in $A_2 \cup A_4$ there exists a tree that
      covers $x_1, x_2, x_3, x_4$.
    \end{claim}

    \begin{proofclaim}
      Up to symmetry $w$ has a neighbor $a_2 \in A_2$. We choose $a_2$
      subject to the minimality of $P_{a_2}$.  Note that here $Q=v=w$.

      If $v$ has also a neighbor $a_3 \in S_3 \cup A_3$ and a neighbor
      $a_4 \in S_4 \cup A_4$ (note that $G$ being triangle-free
      $a_3\in S_3$ and $a_4\in S_4$ cannot happen) then we choose
      $a_3, a_4$ subject to the minimality of respectively $P_{a_3}$,
      $P_{a_4}$. So, $V(P_{a_1} \cup Q \cup P_{a_2} \cup P_{a_3} \cup
      P_{a_4})$ induces a tree that covers $x_1, x_2, x_3, x_4$.
      Hence we may assume that $v$ has no neighbor in $S_4 \cup A_4$.

      If $v$ has a neighbor $a_3 \in S_3 \cup A_3$ then we pick
      $s_3\in S_3$ (if $a_3\in S_3$, we choose $s_3 = a_3$).  We let
      $T_3$ be a tree of $G[A_3 \cup \{v, s_3\}]$ that covers $v, s_3,
      x_3$.  Note that $T_3$ exists by Theorem~\ref{theo_3term}
      because $G[A_3 \cup \{v, s_3\}]$ is connected.  So, $V(P_{a_1}
      \cup Q \cup P_{a_2} \cup T_3 \cup P_{s_4})$ where $s_4\in
      S_4$ induces a tree that covers $x_1, x_2, x_3, x_4$.  Hence we
      may assume that $v$ has no neighbor in $S_3 \cup A_3$.

      Now, we pick $s_1 \in S_1$ and we let $T_1$ be a tree of $G[A_1
        \cup \{v, s_1\}]$ that covers $v, s_1, x_1$.  Note that $T_1$
      exists by Theorem~\ref{theo_3term} because $G[A_1 \cup \{v,
        s_1\}]$ is connected. So, $V(T_1 \cup P_{a_2} \cup
      P_{s_3} \cup P_{s_4})$ where $s_3 \in S_3,s_4\in S_4$ induces a tree that
      covers $x_1, x_2, x_3, x_4$.
   \end{proofclaim}

    So, we may assume that $w$ has no neighbor in $A_2 \cup A_4$.

    \begin{claim}
    If  $w$ has a neighbor in $S_3$ there exists a tree that
      covers $x_1, x_2, x_3, x_4$.
    \end{claim}

    \begin{proofclaim}
      Let $s_3$ be a neighbor of $w$ in $S_3$. Note that $G$ being
      triangle-free, $w$ has no neighbor in $S_2 \cup S_4$.  We let
      $T_3$ be a tree of $G[A_3 \cup \{w, s_3\}]$ that covers $w, s_3,
      x_3$.  Note that in fact $T_3$ is a path either from $s_3$ to
      $x_3$ or from $w$ to $x_3$.  So, $V(P_{a_1} \cup Q \cup P_{s_2}
      \cup T_3 \cup P_{s_4})$ where $s_2 \in S_2$, $s_4 \in S_4$
      induces a tree that covers $x_1, x_2, x_3, x_4$.
    \end{proofclaim}

    So, we may assume that $w$ has no neighbor in $S_3$.

    \begin{claim}
      \label{c:defs2}
      If $w$ has no neighbor in $S_2 \cup S_4$ there exists a tree
      that covers $x_1, x_2, x_3, x_4$.
    \end{claim}

    \begin{proofclaim}
      If $w$ has no neighbor in $S_2 \cup S_4$, by the definition of
      $Q$, $w$ must have a neighbor $a_3 \in A_3$.

      We pick $s_3 \in S_3$. We let $T_3$ be a tree of $G[A_3 \cup
        \{w, s_3\}]$ that covers $w, s_3, x_3$.  So, $V(P_{a_1} \cup Q
      \cup P_{s_2} \cup T_3 \cup P_{s_4})$ where $s_2 \in S_2$, $s_4
      \in S_4$ induces a tree that covers $x_1, x_2, x_3, x_4$.
    \end{proofclaim}

    So, we may assume that $w$ has a neighbor in $S_2 \cup S_4$ (say
    $s_2 \in S_2$ up to symmetry).

    \begin{claim}
      \label{c:defa3}
     If $w$ has no neighbor in $A_3$ then either there exists a tree
     that covers $x_1, x_2, x_3, x_4$ or $N_{S} (w) = S_2 \cup S_4$
     and $N_A(w) \subseteq A_1$.
    \end{claim}

    \begin{proofclaim}
      If $s_4 \in S_4$ is a non-neighbor of $w$, then $V(P_{a_1} \cup
      Q \cup P_{s_2} \cup P_{s_3} \cup P_{s_4})$, where $s_3 \in S_3$
      is a tree that covers $x_1, x_2, x_3, x_4$.  So, we may assume
      that $w$ is complete to $S_4$.  By the same way we may also
      assume that $w$ is complete to $S_2$.  Hence $N_{S} (w) = S_2
      \cup S_4$ and $N_A(w) \subseteq A_1$.
    \end{proofclaim}

    Note that if $N_{S} (w) = S_2 \cup S_4$ and $N_A(w) \subseteq A_1$
    then the second output of our Claim~\ref{c:squareA4} holds.  So,
    from the definition of $Q$ we may assume that $w$ has a neighbor
    $a_3$ in $A_3$.  This implies $v=w$.  We choose $a_3$ subject to
    the minimality of $P_{a_3}$.

    \begin{claim}
      \label{c:defs4}
     If $v$ has a neighbor in $S_4$ there exists a tree that
      covers $x_1, x_2, x_3, x_4$
    \end{claim}

    \begin{proofclaim}
      Let $s_4\in S_4$ be such that $vs_4 \in E(G)$. Then $V(P_{a_1} \cup v
      \cup P_{s_2} \cup P_{a_3} \cup P_{s_4})$ is a
      tree that covers $x_1, x_2, x_3, x_4$.
    \end{proofclaim}

    So, we may assume that $v$ has a non-neighbor $s_4 \in S_4$.

    Let us finish the proof of our claim. We pick $s_1 \in S_1$ and
    $s_3 \in S_3$.  Note that $vs_1\not\in E(G)$ since $G$ is
    triangle-free.  If $a_1s_1 \notin E(G)$ then $V(P_{a_1} \cup P_{s_2} \cup P_{a_3} \cup P_{s_4}) \cup \{s_1,v\}$ induces a
    tree that covers $x_1, x_2, x_3, x_4$. So we may assume $s_1a_1
    \in E(G)$.  Symmetrically, we may assume $s_3a_3 \in E(G)$.

    We observe that $V(P_{a_1} \cup P_{s_2} \cup P_{a_3} \cup P_{s_4})
    \cup \{s_1, s_3, v\}$ is a cubic structure.  A split is given by :
    $A_1 = V(P_{a_1}\sm a_1)$, $A_2 = V(P_{s_2}\sm s_2)$, $A_3 =
    V(P_{a_3}\sm a_3)$, $A_4 = V(P_{s_4}\sm s_4)$, $S_1 = \{a_1\}$
    $S_2 = \{s_2\}$, $S_3 = \{a_3\}$, $S_4 = \{s_4\}$, $S_5 =
    \{s_3\}$, $S_6 = \emptyset$, $S_7 = \{s_1\}$, $S_8 = \{v\}$, $B = \emptyset,R = \emptyset$.
  \end{proof}

  Now, let $C$ be the set of the $(S_2 \cup S_4)$-complete vertices of
  $R\cup \{v\}$.  Let $Y$ be the set of those vertices $w$ of $R \cup
  \{v\}$ such that there exists a path from $v$ to $w$ whose interior
  is in $R$.  Let $Y_1$ be the set of those vertices $w$ of $Y \sm C$
  such that there exists a path from $v$ to $w$ whose interior is in
  $R \sm C$.  Let $Y_2$ be the set of those vertices $w$ of $Y \cap C$
  such that there exists a path from $v$ to $w$ whose interior is in
  $R \sm C$.  Let $Y_3$ be $Y \sm (Y_1 \cup Y_2)$.  Note that $Y = Y_1
  \cup Y_2 \cup Y_3$.

  Note that we may assume that the only possible output of
  Claim~\ref{c:squareA4} is $N_{S} (w) = S_2 \cup S_4$ and $N_A(w)
  \subseteq A_1$.  Also no vertex of $Y$ has a neighbor in $A_2 \cup
  A_3 \cup A_4$ (for $v$ this follows from Claim~\ref{c:squareA4}, for
  the rest of $Y$ this follows from the definition of $R$).  Note that
  $v \notin Y_3$.  But $v\in Y_2$ is possible since $v$ can be
  complete to $S_2 \cup S_4$.  So $N_{Z\cup \{v\}}(Y_3) \subseteq Y_2
  \cup S$ from the definition of $R$.  Also $N_{Z\cup \{v\}}(Y_2)
  \subseteq Y_1 \cup Y_3 \cup A_1 \cup S_2\cup S_4$.  And from
  Claim~\ref{c:squareA4}, $N_{Z\cup \{v\}}(Y_1) \subseteq Y_2 \cup A_1
  \cup S_1$.

  Hence, we can put all the vertices of $Y_1$ in $A_1$, all the
  vertices of $Y_2$ in $S_1$ and leave all the vertices of $Y_3$ in
  $R$. More formally we let:
  \begin{itemize}
\item $A'_1=A_1\cup Y_1$;
\item $A'_i=A_i,\  i=2,3,4$;
\item $S'_1= S_1\cup Y_2$;
\item $S'_i=S_i,\  i=2,3,4$;
\item $R'=R\setminus(Y_1\cup Y_2)$.
\end{itemize}

    We see that $(A'_1,\ldots,A'_4,S'_1,\ldots,S'_4,R')$ 
 is a square structure of $Z\cup \{v\}$.\\

  Here is how to transform the proof above into an algorithm. We first
  compute $C$. After, we use BFS to compute $Y$.  The output of BFS is
  a rooted tree whose root is $v$.  Similarly, we compute $Y_1, Y_2,
  Y_3$.  We check whether $N_{Z\cup \{v\}}(Y_1) \subseteq Y_2 \cup A_1
  \cup S_1$.  If this is true, the paragraph above shows how to output an
  augmented square structure.  Else there is a vertex $w\in Y_1$ such
  that $w$ has neighbors in $(A\sm A_1) \cup (S\sm S_1)$.  Hence by
  backtracking the BFS tree from $w$, we find a path $Q= v \tp \cdots
  \tp w$ where $Q\sm v \subseteq R$ and such that $w$ has neighbors in
  $(A\sm A_1) \cup (S\sm S_1)$.  Moreover, the condition $N_{S} (w)
  = S_2 \cup S_4$ and $N_A(w) \subseteq A_1$ fails since $w\notin C$.
  So the proof of Claim~\ref{c:squareA4} is a description of how, by
  just checking several neighborhoods, we can find either:

  \begin{itemize}
  \item
    a tree of $G[Z \cup \{v\}]$ that covers $x_1,$ $x_2,$ $x_3,$
    $x_4$ or
  \item
    a split of the a cubic structure of $G[Z \cup \{v\}]$,
  \end{itemize}

  This completes the proof of Lemma~\ref{l:square}.

\section{Proof of Lemma~\ref{l:cube}}
\label{s:cube}
  Let $Z \subseteq V(G)$ be a cubic structure of $G$ with respect to
  $x_1, x_2, x_3, x_4$ together with a split like in the definition
  and let $v$ be in $V(G) \sm Z$.  Note that $Z$, the split of $Z$ and
  $v$ are given by assumption.

  Here below, we give a proof of the existence of the objects that the
  algorithm of our Lemma must output, namely a tree or an augmented
  cubic structure.  But like in the proof of Lemma~\ref{l:square},
  this proof is in fact the description of an $O(m)$-time algorithm.
  We omit the details of how to tranform the proof into an algorithm,
  since they are similar to those of the proof of
  Lemma~\ref{l:square}.

  When $s_i \in S_i\cup A_i$, we define the path $P_{s_i}$ to be a
  path from $s_i$ to $x_i$, whose interior is in $A_i$, $i= 1, 2, 3,
  4$.  Note that $P_{s_i}$ exists since by Item~\ref{i:cubeACon} of
  the definition of cubic structures, $G[A_i]$ is connected and by
  Item~\ref{i:cubeNAi} every vertex of $S_i$ has a neighbor in $A_i$,
  $i=1, 2, 3, 4$.

  \begin{CL}
    \label{c:cubeAi}
    The lemma holds when $v$ has neighbors in $A$.
  \end{CL}

  \begin{proof}
    For suppose that $v$ has a neighbor in $A$, say $a_1 \in A_1$ (the
    cases with a neighbor in $A_2, A_3, A_4$ are symmetric).  We
    chooose $a_1$ subject to the minimality of $P_{a_1}$.

    \begin{claim}
      \label{c:cubeAS8}
      If there exists a path $Q= v \tp \cdots \tp w$ where $Q\sm v
      \subseteq B \cup R$ and such that $w$ has neighbors in $A_2
      \cup A_3 \cup A_4 \cup S_2 \cup S_3 \cup S_4 \cup S_5$
      then there exists a tree of $G[Z \cup \{v\}]$ that covers $x_1,$
      $x_2,$ $x_3,$ $x_4$.
    \end{claim}

    \begin{proofclaim}
      Let $Q$ be such a path, minimal with respect to its
      properties.  Note that possibly $v=w$.

      If $w$ is adjacent to $a_2 \in S_2 \cup A_2$, $a_3 \in S_3 \cup
      A_3$ and $a_4 \in S_4 \cup A_4$ then we choose $a_2, a_3, a_4$
      subject to the minimality of $P_{a_2}, P_{a_3}, P_{a_4}$.  So,
      $V(P_{a_1} \cup Q \cup P_{a_2} \cup P_{a_3} \cup P_{a_4})$
      induces a tree of $G$ that covers $x_1, x_2, x_3, x_4$.  Hence,
      by symmetry, we may assume that $w$ has no neighbor in $S_4 \cup
      A_4$.

      If $w$ is adjacent to $a_2 \in S_2 \cup A_2$, $a_3 \in S_3 \cup
      A_3$ then $v=w$ because no vertex in $B\cup R$ can have
      neighbors in both $S_2\cup A_2$, $S_3 \cup A_3$ by
      Items~\ref{i:cubeNBi}, \ref{i:cubeNR}.  We suppose that $a_2,
      a_3$ are chosen subject to the minimality of $P_{a_2}, P_{a_3}$.
      Let $T_2$ be a tree of $G[A_2\cup \{v, s_2\}]$ that covers $x_2,
      v, s_2$ where $s_2$ is some vertex of $S_2$ (if $a_2 \in S_2$ we
      choose $s_2=a_2$).  Note that $T_2$ exists by
      Theorem~\ref{theo_3term} because $G[A_2 \cup \{v, s_2\}]$ is
      connected.  One of $S_6, S_7$ is non-empty by
      Item~\ref{i:cubeNe} of the definition, and we may assume $S_7
      \neq \emptyset$ because of the symmmetry between $S_2, S_7$ and
      $S_3, S_6$.  So, $V(P_{a_1} \cup T_2 \cup P_{a_3} \cup
      P_{s_4}) \cup \{s_7\}$ where $s_4 \in S_4$, $s_7 \in S_7$ is a
      tree that covers $x_1, x_2, x_3, x_4$, except when $vs_7 \in
      E(G)$.  But then, $V(P_{a_1} \cup Q \cup P_{a_2} \cup P_{a_3}
      \cup P_{s_4}) \cup \{s_7\}$ is tree that covers $x_1, x_2, x_3,
      x_4$, because $a_2 \in S_2$ would entail the triangle $a_2s_7w$.
      Hence, by symmetry, we may assume that $w$ has no neighbor in
      $S_3 \cup A_3$.

      If $w$ is adjacent to $a_2 \in S_2 \cup A_2$ then chose $a_2$ subject to the minimality of $P_{a_2}$.  Suppose first
      that some vertex of $Q$ has a neighbor $s_6 \in S_6$.  Then
      $G[A_1 \cup Q \cup S_2 \cup A_2 \cup \{s_6\}]$ is connected, so
      it contains a tree $T_6$ that covers $x_1, x_2, s_6$.  We
      observe that $V(T_6 \cup P_{s_3} \cup P_{s_4})$ where $s_3 \in
      S_3$, $s_4 \in S_4$ is a tree that covers $x_1, x_2, x_3, x_4$.
      Hence we assume from here on that no vertex of $Q$ has a
      neighbor in $S_6$.  Let $T_2$ be a tree of $G[A_2\cup \{s_2,
        w\}]$ that covers $x_2, w, s_2$ where $s_2$ is some vertex of
      $S_2$ (if $a_2 \in S_2$ we choose $s_2=a_2$).  Suppose now that
      $S_5\neq \emptyset$.  We observe that $V(P_{a_1} \cup Q \cup T_2
      \cup P_{s_3} \cup P_{s_4})\cup \{s_5\}$ where $s_3\in S_3$, $s_4
      \in S_4$, $s_{5}\in S_5$, is a tree of $G$ that covers $x_1,
      x_2, x_3, x_4$ except when $ws_{5} \in E(G)$.  But in this case
      we observe that $V(P_{a_1} \cup Q \cup P_{a_2} \cup P_{s_3} \cup
      P_{s_4})\cup \{s_5\}$ is a tree of $G$ that covers $x_1, x_2,
      x_3, x_4$.  Hence, we may assume that $S_5 = \emptyset$ and by
      Item~\ref{i:cubeNe} of the definition we have $S_6, S_7, S_8
      \neq \emptyset$.  If no vertex of $Q$ has a neighbor in $S_7
      \cup S_8$ then $V(P_{a_1} \cup Q \cup T_2 \cup P_{s_3} \cup
      P_{s_4})\cup \{s_7, s_8\}$ where $s_7 \in S_7$, $s_{8}\in S_8$,
      is a tree of $G$ that covers $x_1, x_2, x_3, x_4$.  So we may
      assume that some vertex of $Q$ has a neighbor in $S_7 \cup S_8$
      and we let $u$ be the vertex of $Q$ closest to $v$ that has one
      neighbor in $S_7 \cup S_8$, say $s_7\in S_7$ (the case with one
      neighbor in $S_8$ is similar because of the symmetry between
      $S_7, S_4$ and $S_8, S_3$).  Let $s_2\in S_2$.  So $V(P_{a_1}
      \cup (v \tp Q \tp u) \cup P_{s_2} \cup P_{s_3} \cup P_{s_4})
      \cup \{s_6, s_7\}$ is a tree of $G$ that covers $x_1, x_2, x_3,
      x_4$ except when some vertex of $v \tp Q \tp u$ has a neighbor
      in $P_{s_2}$. But then, by the minimality of $Q$, we must have
      $u=w$.  Now since $G$ is triangle-free, $a_2\notin S_2$. So,
      $V(P_{a_1} \cup Q \cup P_{a_2} \cup P_{s_3} \cup P_{s_4}) \cup
      \{s_6, s_7\}$ is a tree of $G$ that covers $x_1, x_2, x_3, x_4$.
      Hence we may assume that $w$ has no neighbor in $S_2 \cup A_2$.

      Now $w$ has no neighbors in $S_2 \cup S_3 \cup S_4 \cup A_2 \cup
      A_3 \cup A_4$.  So $w$ must have a neighbor $s_5 \in S_5$.
      Hence, $V(P_{a_1} \cup Q \cup P_{s_2} \cup P_{s_3} \cup P_{s_4})
      \cup \{s_5\}$ where $s_2 \in S_2$, $s_3 \in S_3$, $s_4 \in S_4$
      is a tree that covers $x_1, x_2, x_3, x_4$.
    \end{proofclaim}

    \begin{claim}
      \label{c:cubeAComp}
      If there exists a path $Q= v \tp \cdots \tp w$ whose interior is
      in $B \cup R$ and such that $w$ has neighbors in $S_6 \cup S_7
      \cup S_8$ then either $Q$ contains a vertex that is complete to
      $S_6 \cup S_7 \cup S_8$ or there exists a tree of $G[Z \cup
        \{v\}]$ that covers $x_1, x_2, x_3, x_4$.
    \end{claim}

    \begin{proofclaim}
      Let $Q$ be such a minimal path.  It suffices to prove that
      $w$ is complete to $S_6 \cup S_7 \cup S_8$ or that a tree
      covering $x_1, x_2, x_3, x_4$ exists.  By~(\ref{c:cubeAS8}), we
      may assume that no vertex of $Q$ has a neighbor in $A_2 \cup A_3
      \cup A_4 \cup S_2 \cup S_3 \cup S_4 \cup S_5$.  So up to the
      symmetry between $S_6, S_7, S_8$, we may assume that $w$ has a
      non-neighbor $s_6 \in S_6$ and a neighbor $s_7 \in S_7$ for
      otherwise our claim is proved.  Hence $V(P_{a_1} \cup Q \cup
      P_{s_2} \cup P_{s_3} \cup P_{s_4}) \cup \{s_6, s_7\}$ is a tree
      that covers $x_1, x_2, x_3, x_4$.
    \end{proofclaim}

    Now, let $C$ be the set of the $(S_6 \cup S_7 \cup S_8)$-complete
    vertices of $Z\cup \{v\}$.  Let $Y$ be the set of these vertice
    $w$ of $B \cup R \cup \{v\}$ such that there exists a path from
    $v$ to $w$ whose interior is in $B \cup R$.  Let $Y_1$ be the set
    of these vertices $w$ of $Y \sm C$ such that there exists a path from
    $v$ to $w$ whose interior is in $(B \cup R) \sm C$.  Let $Y_2$ be
    the set of these vertices $w$ of $Y\cap C$ such that there exists a
    path from $v$ to $w$ whose interior is in $(B \cup R)\sm C$.  Let
    $Y_3$ be $Y \sm (Y_1 \cup Y_2)$.  Note that $Y = Y_1 \cup Y_2 \cup
    Y_3$.

    By~(\ref{c:cubeAS8}), we may assume that no vertex of $Y$ has a
    neighbor in $A_2 \cup A_3 \cup A_4 \cup S_2 \cup S_3 \cup S_4 \cup
    S_5$.  Note that $v \notin Y_3$.  But $v\in Y_2$ is possible since
    $v$ can be complete to $(S_6 \cup S_7 \cup S_8)$.  So
    by~(\ref{c:cubeAComp}), $N_{Z\cup \{v\}}(Y_3) \subseteq Y_2 \cup
    S_1$.  Also by~(\ref{c:cubeAComp}),
    $N_{Z\cup \{v\}}(Y_2) \subseteq Y_1 \cup Y_3 \cup A_1 \cup S_6 \cup S_7 \cup S_8$. And
    $N_{Z\cup \{v\}}(Y_1) \subseteq Y_2 \cup A_1 \cup S_1$.

    Hence, we can put all the vertices of $Y_1$ in $A_1$, all the
    vertices of $Y_2$ in $S_1$ and all the vertices of $Y_3$ in $B_1$.
    We obtain a split of the cubic structure $Z\cup \{v\}$.
  \end{proof}

  \begin{CL}
    \label{c:cubevCS}
    The lemma holds if $v$ is complete to $(S_1 \cup S_2 \cup S_3 \cup
    S_4) \sm S_i$, $i=1, 2, 3, 4$.
  \end{CL}

  \begin{proof}
    We prove the claim when $i=4$, the other cases are symmetric.  So
    $v$ is complete to $S_1 \cup S_2 \cup S_3$.

    \begin{claim}
      \label{c:cubeS4}
      If there exists a path $Q = v \tp \cdots \tp w$ such that $V(Q
      \sm v) \subseteq B \cup R$ and $w$ has a neighbor $s_4$ in $S_4$
      then $G[Z \cup \{v\}]$ contains a tree that covers $x_1, x_2,
      x_3, x_4$.
    \end{claim}

    \begin{proofclaim}
      Let us consider such a path $Q$ minimal with respect its
      properties.  Every vertex of $Q \sm v$ is in $B_4$.  Indeed,
      $B_4$ is the only set among $B_1, \dots, B_4, R$ that allows
      neighbors in $S_4$, and there are no edges between the sets
      $B_1, \dots, B_4, R$.  Hence by the properties of $B_4$, no
      vertex of $Q \sm v$ can have neighbors in $S_1 \cup S_2 \cup
      S_3$.  So, $V(P_{s_4} \cup Q \cup P_{s_1} \cup P_{s_2} \cup
      P_{s_3})$ where $s_1 \in S_1$, $s_2 \in S_2$, $s_3 \in S_3$,
      induces a tree that covers $x_1, x_2, x_3, x_4$.
    \end{proofclaim}

    Let $Y$ be the set of these vertices $w$ of $B \cup R$ such that
    there exists a path $Q = v \tp \cdots \tp w$ whose interior is in
    $B\cup R$.  If $v$ has some neighbors in $B_4$ then
    by~(\ref{c:cubeS4}) we may assume that every component of $G[Y
      \cap B_4]$ contains no neighbors of vertices of $S_4$.  So,
    every such component can be taken out of $B_4$ and put in $R$
    instead.  Then we may put $v$ in $S_8$ and we obtain a split of
    the cubic structure $Z \cup \{v\}$.
  \end{proof}

  \begin{CL}
    \label{c:cubeBS}
    The lemma holds.
  \end{CL}

  \begin{proof}
    By Claim~\ref{c:cubeAi} we may assume that $v$ has no neighbor in
    $A$.

    \begin{claim}
      \label{c:cubeS123}
      For the pairs $(i, j)$ such that $1 \leq i < j \leq 4$ and the
      pairs $(i, j)$ among $(1, 5)$, $(2, 6)$, $(3, 7)$, $(4, 8)$ the
      statement below is true:

      If there exists a path $Q = u \tp \cdots \tp w$ of $G[B \cup R
        \cup \{v\}]$ such that $u$ has a neighbor
      in $S_i$ and $w$ has a neighbor in $S_j$ then the lemma holds.
    \end{claim}

    \begin{proofclaim}
      Let us choose such a pair $(i, j)$ and such a path $Q$, subject
      to the minimality of $Q$. Note that by the definition of a cubic structure, $V(Q)\subseteq B\cup R$ is impossible. So, $Q$ contains $v$.

      If $u$ is adjacent to $s_1 \in S_1, s_2\in S_2, s_3 \in S_3, s_4
      \in S_4$ then $Q=u=v$ by the minimality of~$Q$. So, $V(P_{s_1}
      \cup P_{s_2} \cup P_{s_3} \cup P_{s_4} \cup Q)$ induces a tree
      that covers $x_1, x_2, x_3, x_4$.  Hence, we may assume that $u$
      (and symmetrically $w$) has no neighbor in $S_4$.

      If $u$ is adjacent to $s_1 \in S_1, s_2\in S_2, s_3 \in S_3$
      then $Q =u = v$ by the minimality of $Q$.  By
      Claim~\ref{c:cubevCS} we may assume that $v$ is not complete to
      $S_1 \cup S_2 \cup S_3$, so $v$ has a non-neighbor in $S_1 \cup
      S_2 \cup S_3$, say $s_1 \in S_1$.  Let $s_2 \in S_2, s_3 \in
      S_3$ be neighbors of $v$.  By Item~\ref{i:cubeNe} of the
      definition, we have $S_6 \cup S_7 \neq \emptyset$, so up to the
      symmetry between $S_2, S_7$ and $S_3, S_6$ we may assume that
      there exists $s_6\in S_6$.  Note that $vs_6 \notin E(G)$
      because $G$ is triangle-free.  So $V(P_{s_1} \cup P_{s_2} \cup
      P_{s_3} \cup P_{s_4} \cup Q) \cup \{s_6\}$ induces a tree that
      covers $x_1, x_2, x_3, x_4$.  So, we may assume that $u$ (and
      symmetrically $w$) has no neighbor in $S_3$.

      If $(i, j)$ is such that $1\leq i < j \leq 4$ then up to
      symmetry, $u$ has a neighbor in $s_1 \in S_1$ and $w$ has a
      neighbor $s_2 \in S_2$.  No vertex of $Q$ has neighbors in $S_5
      \cup S_6$ because such a vertex would form a triangle or would
      contradict the minimality of $Q$.  Also no vertex of $Q$ has
      neighbors in $S_3 \cup S_4$.  Indeed for $u, w$ this follows
      from the preceeding paragraphs, and for the interior vertices of
      $Q$, it follows from the minimality of $Q$.  So, $V(P_{s_1} \cup
      P_{s_2} \cup P_{s_3} \cup P_{s_4} \cup Q) \cup \{s\}$ where
      $s\in S_5\cup S_6$ is a tree that covers $x_1, x_2, x_3, x_4$.

      If $u$ has a neighbor $s_1 \in S_1$ and $w$ has a neighbor $s_5
      \in S_5$, then no vertex of $Q$ has neighbors in $S_2 \cup S_3
      \cup S_4$.  Indeed, such a vertex would form a triangle or would
      contradict the minimality of $Q$.  So, $V(P_{s_1} \cup P_{s_2}
      \cup P_{s_3} \cup P_{s_4} \cup Q) \cup \{s_5\}$ induces a tree
      that covers $x_1, x_2, x_3, x_4$.  Similarly, we can prove that
      our claim holds when $(i, j)$ is one of $(2, 6)$, $(3, 7)$, $(4,
      8)$.
    \end{proofclaim}

    Let $Y$ be the set of these vertice $u$ of $B \cup R \cup \{v\}$
    such that there exists a path from $v$ to $u$ whose interior is in
    $B \cup R$.  From~(\ref{c:cubeS123}) it follows that $N_Z(Y)$ is
    included in either $S_1 \cup N_S(S_1)$, $S_2 \cup N_S(S_2)$, $S_3
    \cup N_S(S_3)$, $S_4 \cup N_S(S_4)$ or $S_5 \cup S_6 \cup S_7 \cup
    S_8$.  So, respectively to these cases, $Y$ can be put in either
    $B_1$, $B_2$, $B_3$, $B_4$ or $R$, and we obtain a split of the
    cubic structure $Z \cup \{v\}$.
 \end{proof}

This completes the proof of Lemma~\ref{l:cube}.

\section{NP-completeness of four-in-a-centered-tree}
\label{s:npc}

The NP-completeness of four-in-a-centered-tree follows directly from
the fact (proved by Bienstock~\cite{bienstock:evenpair}) that the
problem of detecting an induced cycle passing through two prescribed
vertices of a graph is NP-complete.  In fact, the NP-completeness
result of Bienstock remains true for several classes of graphs where
some induced subgraphs are forbidden.  In~\cite{leveque.lmt:detect},
L\'ev\^eque, Lin, Maffray and Trotignon study the kinds of graph that
can be forbidden.  We use one of their result.  When $k\geq 3$, we
denote by $C_k$ the cycle on $k$ vertices.

\begin{theorem}[see \cite{leveque.lmt:detect}]
  \label{th:npLin}
  Let $k\geq 3$ be an integer.  Then the following problem is
  NP-complete:
  \begin{description}
  \item
    {\sc Instance:} two vertices $x,y$ of degree~2 of a graph $G$ that
    does not contain $C_3, \dots, C_k$.
  \item
    {\sc Question:} does $G$ contain an induced cycle covering $x,
    y$~?
  \end{description}
\end{theorem}

We deduce easily:

\begin{theorem}
  Let $k\geq 3$ be an integer.  Then the following problem is
  NP-complete:
  \begin{description}
  \item
    {\sc Instance:} four terminals $x_1, x_2, x_3, x_4$ of a graph $G$
    that does not contain $C_3, \dots, C_k$.
  \item
    {\sc Question:} Does $G$ contain a centered tree covering $x_1,
    x_2, x_3, x_4$~?
  \end{description}
\end{theorem}

\begin{proof}
  Let us consider an instance $G, x, y$ of the NP-complete problem of
  Theorem~\ref{th:npLin}.  Let $x', x''$ be the neighbors of $x$ and
  $y', y''$ be the neighbors of $y$. We prepare an instance $G', x_1,
  x_2, x_3, x_4$ of our problem as follows. We delete $x, y$.  We add
  five vertices $c, x_1, x_2, x_3, x_4$ and the following edges:
  $cx_1, cx_2, cx', cx'', x_3y', x_4y''$.  Now, $G', x_1, x_2, x_3,
  x_4$ is an instance of our problem.

  Since $x_1 \tp c \tp x_2$ is a $P_3$ of $G'$ and since $x_1, x_2$
  are of degree~1, every induced centered tree of $G'$ covering $x_1,
  x_2, x_3, x_4$ must be made of four edge-disjoint paths $c \tp x_1$,
  $c \tp x_2$, $c\tp \cdots \tp x_3$, $c\tp \cdots \tp x_4$.  So, such
  a tree exists if and only if there exists an induced cycle of $G$
  covering $x, y$. This proves that our problem is NP-complete.
\end{proof}

By the same way, four-in-a-centered-tree can be proved NP-complete for
several classes of graphs defined by a given list $\cal L$ of
forbidden subgraphs.  Each time, the proof relies on a direct
application of an NP-completeness result for $\cal L$-free graphs
from~\cite{leveque.lmt:detect}.  Going into the details of every
possible list that we can get would not be too illuminating since the
lists are described in~\cite{leveque.lmt:detect}.  Let us just mention
one result: four-in-a-centered-tree is NP-complete for triangle-free
graphs with every vertex except one of degree at most three.

\section*{Acknowledgement}

We are grateful to Paul Seymour for pointing out to us a simplification
in our original NP-completeness proof.

The authors are also grateful to the anonymous referee for his helpful
comments which helped in improving the quality of this article.

\end{document}